\documentclass[a4,a4wide]{article}
\usepackage{aaai}
\usepackage{nmr2014}

\newtheorem{theorem}{Theorem}
\newtheorem{lemma}[theorem]{Lemma}
\newtheorem{definition}[theorem]{Definition}
\newenvironment{proof}{\emph{Proof.}}{\smallskip}

\def\squareforqed{\hbox{\rlap{$\sqcap$}$\sqcup$}}
\def\qed{\ifmmode\squareforqed\else{\unskip\nobreak\hfil
\penalty50\hskip1em\null\nobreak\hfil\squareforqed
\parfillskip=0pt\finalhyphendemerits=0\endgraf}\fi}

\begin{document}
%\pagenumbering{gobble}
%\nocopyrightcommand
%\copyrighttext{}
\nocopyright

\title{Tableau vs.\ Sequent Calculi for Minimal Entailment}

\author{Olaf Beyersdorff\thanks{Supported by a grant from the John Templeton Foundation.} \and Leroy Chew\thanks{Supported by a Doctoral Training Grant from EPSRC.}\\
School of Computing, University of Leeds, UK
  %\email{o.beyersdorff@leeds.ac.uk}
  }

\maketitle

\begin{abstract}
In this paper we compare two proof systems for minimal entailment: a tableau system \OTAB and a sequent calculus \MLK, both developed by Olivetti (1992). Our main result shows that \OTAB-proofs can be efficiently translated into \MLK-proofs, \ie \MLK p-simulates \OTAB. The simulation is technically very involved and answers an open question posed by Olivetti (1992) on the relation between the two calculi. We also show that the two systems are exponentially separated, \ie there are formulas which have polynomial-size \MLK-proofs, but require exponential-size \OTAB-proofs.  
\end{abstract}

\section{Introduction}

Minimal entailment is the most important special case of circumscription, which in turn is one of the main formalisms for non-monotonic reasoning \cite{McC80}. The key intuition behind minimal entailment is the notion of minimal models, providing as few exceptions as possible. Apart from its foundational relation to human reasoning, minimal entailment has wide-spread applications, e.g.\ in AI, description logics \cite{BLW09,GH09,GGOP13} and SAT solving \cite{JM11}. %There has been recent interest for tableau systems for description logics including work on circumscription such as in \cite{GH09} and \cite{GGOP13}.

While the complexity of non-monotonic logics has been thoroughly studied --- cf.\ e.g.\ the recent papers  \cite{DHN12,Tho12,BLW09} or the survey \cite{TV10} --- considerably less is known about the complexity of theorem proving in these logics. This is despite the fact that a number of quite different formalisms have been introduced for circumscription and minimal entailment \cite{Oli92,Nie96,BO02,GH09,GGOP13}. While proof complexity has traditionally focused on proof systems for classical propositional logic, there has been remarkable interest in proof complexity of non-classical logics during the last decade. A number of exciting results have been obtained --- in particular for modal and intuitionistic logics \cite{Hru09,Jer09} ---  and interesting phenomena have been observed that show a quite different picture from classical proof complexity, cf.\ \cite{BK12} for a survey. 

In this paper we focus our attention at two very different formalisms for minimal entailment: a sequent calculus \MLK and a tableau system \OTAB, both developed by Olivetti \shortcite
{Oli92}.\footnote{While the name \MLK is Olivetti's original notation \cite{Oli92}, we introduce the name \OTAB here as shorthand for Olivetti's tableau. By \NTAB we denote another tableau for minimal entailment suggested by \Niemela \shortcite{Nie96}, cf.\ the conclusion of this paper.
}  These systems are very natural and elegant, and in fact they were both inspired by their classical propositional counterparts: Gentzen's \LK \shortcite{Gen35} and Smullyan's analytic tableau \shortcite{Smu68}. 

Our main contribution is to show a p-simulation of \OTAB by \MLK, \ie proofs in \OTAB can be efficiently transformed into \MLK-derivations. This answers an open question by Olivetti \shortcite{Oli92} on the relationship between these two calculi. At first sight, our result might not appear unexpected as sequent calculi are usually stronger than tableau systems, cf. e.g. \cite{Urq95}. However, the situation is more complicated here, and even Olivetti himself did not seem to have a clear conjecture as to whether such a simulation should be expected, cf.\ the remark after Theorem~8 in \cite{Oli92}.

The reason for the complication lies in the nature of the tableau: while rules in \MLK are `local', \ie they refer to only two previous sequents in the proof, the conditions to close branches in \OTAB are `global' as they refer to other branches in the tableau, and this reference is even recursive. The trick we use to overcome this difficulty is to annotate nodes in the tableau with additional information that `localises' the global information. This annotation is possible in polynomial time. The annotated nodes are then translated into minimal entailment sequents that form the skeleton of the \MLK derivation for the p-simulation. 

In addition to the p-simulation of \OTAB by \MLK, we obtain an exponential separation between the two systems, \ie there are formulas which have polynomial-size proofs in \MLK, but require exponential-size \OTAB tableaux. In proof complexity, lower bounds and separations are usually much harder to show than simulations, and indeed there are famous examples where simulations have been known for a long time, but separations are currently out of reach, cf.\ \cite{Kra95}. In contrast, the situation is opposite here: while the separation carries over rather straightforwardly from the comparison between classical tableau and \LK, the proof of the simulation result is technically very involved.

This paper is organised as follows. We start by recalling basic definitions from minimal entailment and proof complexity, and explaining Olivetti's systems \MLK and \OTAB for minimal entailment \cite{Oli92}. This is followed by two sections containing the p-simulation and the separation of \OTAB and \MLK. In the last section, we conclude  by placing our results into the global picture of proof complexity research on circumscription and non-monotonic logics.

%\ref{sec:prelim}
%\ref{sec:def-MLK-OTAB}
%\ref{sec:simulation-OTAB-MLK}
%\ref{sec:separation-OTAB-MLK}
%\ref{sec:conclusion}

\section{Preliminaries
\label{sec:prelim}}

Our propositional language contains the logical symbols $ \bot,\top,\neg,\vee,\wedge,\rightarrow$. 
For a set of formulae $\Sigma$, $\VAR{(\Sigma)}$ is the set of all atoms that occur in $\Sigma$. For a set $P$ of atoms we set $\neg P =\{\neg p \mid p\in P\}$. Disjoint union of two sets $A$ and $B$ is denoted by $A\sqcup B$.

\subsubsection{Minimal Entailment.}

Minimal entailment is a form of non-monotonic reasoning developed as a special case of McCarthy's circumscription \cite{McC80}. Minimal entailment comes both in a propositional and a first-order variant. Here we consider only the version of minimal entailment for propositional logic. We identify models with sets of positive atoms and use the partial ordering $\subseteq$ based on inclusion. This gives rise to a natural notion of minimal model for a set of formulae, in which the number of positive atoms is minimised with respect to inclusion. For a set of propositional formulae $\Gamma$ we say that $\Gamma$ minimally entails a formula $\phi$ if all minimal models of $\Gamma$ also satisfies $\phi$. We denote this entailment by $\Gamma \vDash_M \phi$.

\subsubsection{Proof Complexity.}
A \emph{proof system} \cite{CR79} for a language $L$ over alphabet $\Gamma$ is a polynomial-time computable partial function $f:\Gamma^\star\rightharpoondown\Gamma^\star$ with $\mathit{rng}(f)=L$.
%$\{y\in\Gamma^\star \mid \exists x\in \Gamma^\star \mbox{ s.t. } f(x)=y \}=L$.
An \emph{$f$-proof} of string $y$ is a string $x$ such that $f(x)=y$.

%From this we can start defining proof size. For $f$ a proof system for language $L$ and string $x\in L$ we define $s_{f}(x)=\min(|w| : f(w)=x)$. Thus the partial function $s_{f}$ tells us the minimum proof size of a theorem. We can overload the notation by setting $s_{f}(n)=\max(s_{f}(x) : |x| \leq n)$ where $ n \in \mathbb{N} $. For a function $t: \mathbb{N} \rightarrow\mathbb{N}$, a proof system $f$ is called \emph{$t$-bounded} if $\forall n\in \mathbb{N} , s_{f}(n) \leq t(n)$. 

Proof systems are compared by simulations. We say that a proof system $f$ \emph{simulates} $g$ ($g\leq f$) if
there exists a polynomial $p$ such that for every $g$-proof $\pi_g$
there is an $f$-proof $\pi_f$ with $f(\pi_f)=g(\pi_g)$ and
$\abs{\pi_f}\leq p(\abs{\pi_g})$. If $\pi_f$ can even be constructed
from $\pi_g$ in polynomial time, then we say that $f$
\emph{p-simulates} $g$ ($g\leq_p f$). 
Two proof systems $f$ and $g$ are \emph{(p-)equivalent} ($g\equiv_{(p)} f$) if they mutually
(p-)simulate each other.

The sequent calculus of \emph{Gentzen's system \LK} is one of the historically first and best studied proof systems \cite{Gen35}. In \LK a sequent is usually written in the form $ \Gamma \vdash \Delta$. Formally, a \emph{sequent} is a pair ($\Gamma$,$\Delta$) with $\Gamma$ and $\Delta$ finite sets of formulae. In classical logic $ \Gamma \vdash \Delta$ is true if every model for $\bigwedge\Gamma$ is also a model of $\bigvee\Delta$, where the disjunction of the empty set is taken as $\bot$ and the conjunction as $\top$.
The system can be used both for propositional and first-order logic; the propositional rules are displayed in Fig.~\ref{fig_LK}. 
Notice that the rules here do not contain structural rules for contraction or exchange. These come for free as we chose to operate with sets of formulae rather than sequences.
Note the soundness of rule ($\bullet\vdash$), which gives us monotonicity of classical propositional logic. 

\begin{figure}[h]
\framebox{\parbox{\breite}
{%\small
\begin{prooftree}
\AxiomC{}
\RightLabel{($\vdash$)}
\UnaryInfC{$A\vdash A$}
\DisplayProof\hspace{0.5cm}
\AxiomC{}
\RightLabel{($\bot\vdash$)}
\UnaryInfC{$\bot\vdash$}
\DisplayProof\hspace{0.5cm}
\AxiomC{}
\RightLabel{($\vdash\top$)}
\UnaryInfC{$\vdash\top$}
\end{prooftree}
\begin{prooftree}
  \AxiomC{$\Gamma\vdash\Sigma$}
  \RightLabel{($\bullet\vdash$)}
  \UnaryInfC{$\Delta,\Gamma\vdash\Sigma$}
  \DisplayProof\hspace{1cm}
  \AxiomC{$\Gamma\vdash\Sigma$}
  \RightLabel{($\vdash\bullet$)}
  \UnaryInfC{$\Gamma\vdash\Sigma,\Delta$}
%  \DisplayProof\hspace{0.7cm}
\end{prooftree}
\begin{prooftree}
  \AxiomC{$\Gamma\vdash\Sigma,A$}
  \RightLabel{($\neg\vdash$)}
  \UnaryInfC{$\neg A,\Gamma\vdash\Sigma$}
%\end{prooftree}
%\begin{prooftree}
 \DisplayProof\hspace{0.7cm}
  \AxiomC{$A,\Gamma\vdash\Sigma$}
  \RightLabel{($\vdash\neg$)}
  \UnaryInfC{$\Gamma\vdash\Sigma,\neg A$}
%\DisplayProof\hspace{0.7cm}
\end{prooftree}
\begin{prooftree}
  \AxiomC{$A,\Gamma\vdash\Sigma$}
  \RightLabel{($\bullet\wedge\vdash$)}
  \UnaryInfC{$B\wedge A,\Gamma\vdash\Sigma$}
  \DisplayProof\hspace{0.7cm}
  \AxiomC{$A,\Gamma\vdash\Sigma$}
  \RightLabel{($\wedge\bullet\vdash$)}
  \UnaryInfC{$A\wedge B,\Gamma\vdash\Sigma$}
\end{prooftree}
\begin{prooftree}
% \DisplayProof\hspace{0.7cm}
  \AxiomC{$\Gamma\vdash\Sigma,A$}
  \AxiomC{$\Gamma\vdash\Sigma,B$}
  \RightLabel{($\vdash\wedge$)}
  \BinaryInfC{$\Gamma\vdash\Sigma,A\wedge B$}
\end{prooftree}
\begin{prooftree}  
%\DisplayProof\hspace{0.7cm}
  \AxiomC{$A,\Gamma\vdash\Sigma$}
  \AxiomC{$B,\Gamma\vdash\Sigma$}
  \RightLabel{($\vee\vdash$)}
  \BinaryInfC{$A\vee B,\Gamma\vdash\Sigma$}
\end{prooftree}
\begin{prooftree}
%  \DisplayProof\hspace{0.7cm}
  \AxiomC{$\Gamma\vdash\Sigma,A$}
  \RightLabel{($\vdash\bullet\vee$)}
  \UnaryInfC{$\Gamma\vdash\Sigma,B\vee A$}
  \DisplayProof\hspace{0.5cm}
  \AxiomC{$\Gamma\vdash\Sigma,A$}
  \RightLabel{($\vdash\vee\bullet$)}
  \UnaryInfC{$\Gamma\vdash\Sigma,A\vee B$}
\end{prooftree}
\begin{prooftree}
%  \DisplayProof\hspace{0.5cm}
  \AxiomC{$A,\Gamma\vdash\Sigma,B$}
  \RightLabel{($\vdash\rightarrow$)}
  \UnaryInfC{$\Gamma\vdash\Sigma,A \rightarrow B$}
\end{prooftree}
\begin{prooftree}
%  \DisplayProof\hspace{0.7cm}
  \AxiomC{$\Gamma\vdash\Sigma, A$}
  \AxiomC{$B, \Delta\vdash\Lambda$}
  \RightLabel{($\rightarrow\vdash$)}
  \BinaryInfC{$A\rightarrow B, \Gamma,\Delta\vdash\Sigma, \Lambda$}
 % \DisplayProof\hspace{0.7cm}
 \end{prooftree}
\begin{prooftree}
  \AxiomC{$\Gamma\vdash\Sigma,A$}
  \AxiomC{$A,\Gamma\vdash\Sigma$}
  \RightLabel{(cut)}
  \BinaryInfC{$\Gamma\vdash\Sigma$}
\end{prooftree}
\caption{Rules of the sequent calculus \LK \cite{Gen35} }
\label{fig_LK} 
}}
\end{figure}

\section{Olivetti's sequent calculus and tableau system for minimal entailment}
\label{sec:def-MLK-OTAB}

In this section we review two proof systems for minimal entailment, which were developed by Olivetti \shortcite{Oli92}. We start with the sequent calculus \MLK. Semantically, a minimal entailment sequent $\Gamma\vdash_M \Delta$ is true if and only if in all minimals models of $\bigwedge \Gamma$ the formula $\bigvee \Delta$ is satisfied. In addition to all axioms and rules from \LK, the calculus \MLK comprises the axioms and rules detailed in Figure~\ref{fig_MLK}. In the \MLK axiom, the notion of a \emph{positive} atom $p$ in a formula $\phi$ is defined inductively by counting the number of negations and implications in $\phi$ on top of $p$ (cf.\ \cite{Oli92} for the precise definition).

\begin{figure}[h]
\framebox{\parbox{\breite}
{%\small
\begin{prooftree}
\AxiomC{}
\RightLabel{($\vdash_{M}$)}
\UnaryInfC{$\Gamma\vdash_{M}\neg p$}
\end{prooftree}
where $p$ is an atom that does not occur positively in any formula in $\Gamma$
\begin{prooftree}
\AxiomC{$\Gamma\vdash\Delta$}
\RightLabel{($\vdash\vdash_{M}$)}
\UnaryInfC{$\Gamma\vdash_{M}\Delta$}
\end{prooftree}
\begin{prooftree}
\AxiomC{$\Gamma\vdash_{M}\Sigma,A$}
\AxiomC{$A,\Gamma\vdash_{M}\Lambda$}
\RightLabel{(M-cut)}
\BinaryInfC{$\Gamma\vdash_M\Sigma,\Lambda$}
\end{prooftree}
\begin{prooftree}
\AxiomC{$\Gamma\vdash_{M}\Sigma$}
\AxiomC{$\Gamma\vdash_{M}\Delta$}
\RightLabel{($\bullet\vdash_{M}$)}
\BinaryInfC{$\Gamma,\Sigma\vdash_{M}\Delta$}
\end{prooftree}
\begin{prooftree}
  \AxiomC{$\Gamma\vdash_{M}\Sigma,A$}
  \AxiomC{$\Gamma\vdash_{M}\Sigma,B$}
  \RightLabel{($\vdash_{M}\wedge$)}
  \BinaryInfC{$\Gamma\vdash_{M}\Sigma,A\wedge B$}
\end{prooftree}
\begin{prooftree}
  \AxiomC{$A,\Gamma\vdash_{M}\Sigma$}
  \AxiomC{$B,\Gamma\vdash_{M}\Sigma$}
  \RightLabel{($\vee\vdash_{M}$)}
  \BinaryInfC{$A\vee B,\Gamma\vdash_{M}\Sigma$}
\end{prooftree}
\begin{prooftree}
  \AxiomC{$\Gamma\vdash_{M}\Sigma,A$}
  \RightLabel{($\vdash_{M}\bullet\vee$)}
  \UnaryInfC{$\Gamma\vdash_{M}\Sigma,B\vee A$}
\end{prooftree}
\begin{prooftree}
  \AxiomC{$\Gamma\vdash_{M}\Sigma,A$}
  \RightLabel{($\vdash_{M}\vee\bullet$)}
  \UnaryInfC{$\Gamma\vdash_{M}\Sigma,A\vee B$}
\end{prooftree}
\begin{prooftree}
  \AxiomC{$A,\Gamma\vdash_{M}\Sigma$}
  \RightLabel{($\vdash_{M}\neg$)}
  \UnaryInfC{$\Gamma\vdash_{M}\Sigma,\neg A$}
\end{prooftree}
\begin{prooftree}
  \AxiomC{$A,\Gamma\vdash_{M}\Sigma,B$}
  \RightLabel{($\vdash_{M}\rightarrow$)}
  \UnaryInfC{$\Gamma\vdash_{M}\Sigma,A \rightarrow B$}
\end{prooftree}
\caption{  \label{fig_MLK} Rules of the sequent calculus \MLK for minimal entailment \cite{Oli92} }}}
\end{figure}

\begin{theorem} \textbf{(Theorem 8 in \cite{Oli92})} \label{thm_MLK_comp}
A sequent $\Gamma\vdash_{M}\Delta$  is true iff it is derivable in \MLK.
\end{theorem} 

In addition to the sequent calculus \MLK, Olivetti developed a tableau calculus for minimal entailment \cite{Oli92}. Here we will refer to this calculus as \OTAB. A tableau is a rooted tree where nodes are labelled with formulae. 
\setlength{\tabcolsep}{2pt}
\begin{figure}[h]
{%\small
\begin{tabular}{| l | l | l |}
\hline
$\alpha$&$\alpha_1$&$\alpha_2$\\
\hline
$T(A\wedge B)$&$TA$&$TB$\\
$F\neg (A\wedge B)$&$F\neg A$&$F\neg B$\\
$T\neg (A\vee B)$&$T\neg A$&$T\neg B$\\
$F(A\vee B)$&$FA$&$FB$\\
$T\neg (A\rightarrow B)$&$T A$&$T\neg B$\\
$F(A\rightarrow B)$&$F\neg A$&$FB$\\
$T\neg \neg A$&$T A$&$T A$\\
$F\neg \neg A$&$F A$&$F A$\\
\hline
\end{tabular}
\quad
\begin{tabular}{| l | l | l |}
\hline
$\beta$&$\beta_1$&$\beta_2$\\
\hline
$T(A\vee B)$&$TA$&$TB$\\
$F\neg (A\vee B)$&$F\neg A$&$F\neg B$\\
$T\neg (A\wedge B)$&$T\neg A$&$T\neg B$\\
$F(A\wedge B)$&$FA$&$FB$\\
$T(A\rightarrow B)$&$T\neg A$&$T B$\\
$F\neg (A\rightarrow B)$&$F A$&$F\neg B$\\
\hline
\end{tabular}
}
\caption{Classification of signed formulae into $\alpha$ and $\beta$-type by sign and top-most connective \label{fig_contype}}
\end{figure}
In \OTAB, the nodes are labelled with formulae that are signed with the symbol $T$ or $F$. The combination of the sign and the top-most connective allows us to classify signed formulas  into $\alpha$ or $\beta$-type formulae as detailed in Figure~\ref{fig_contype}. Intuitively, for an $\alpha$-type formula, a branch in the tableau is augmented by $\alpha_1, \alpha_2$, whereas for a $\beta$-type formula it splits according to $\beta_1, \beta_2$.  Nodes in the tableau can be either marked or unmarked. For a sequent $\Gamma\vdash_M \Delta$, an \OTAB tableau is constructed by the following process. We start from an initial tableau consisting of a single branch of unmarked formulae, which are exactly all formulae $\gamma\in\Gamma$, signed as $T\gamma$, and all formulae $\delta\in\Delta$, signed as $F\delta$. For a tableau and a branch $\mathcal{B}$ in this tableau we can extend the tableau by two rules:

\begin{itemize}
\item[(A)] If formula $\phi$ is an unmarked node in $\mathcal{B}$ of type $\alpha$, then mark $\phi$ and add the two unmarked nodes $\alpha_1$ and $\alpha_2$ to the branch. 
\item[(B)] If formula $\phi$ is an unmarked node in $\mathcal{B}$ of type $\beta$, then mark $\phi$ and split $\mathcal{B}$ into two branches $\mathcal{B}_1,\mathcal{B}_2$ with unmarked $\beta_1\in\mathcal{B}_1$ and unmarked $\beta_2\in\mathcal{B}_2$.
\end{itemize}

A branch $\mathcal{B}$ is \emph{completed} if and only if all unmarked formulae on the branch are literals.
A branch $\mathcal{B}$ is \emph{closed} if and only if it satisfies at least one of the following conditions:

\begin{enumerate}
\item For some formula $A$, $TA$ and $T\neg A$ are nodes of $\mathcal{B}$ ($T$-closed).
\item For some formula $A$, $FA$ and $F\neg A$ are nodes of $\mathcal{B}$ ($F$-closed).
\item For some formula $A$, $TA$ and $FA$ are nodes of $\mathcal{B}$ ($TF$-closed).
\end{enumerate}

For branch $\mathcal{B}$ let $\mathrm{At}(\mathcal{B})=\{p:p $ is an atom and $Tp$ is a node in $\mathcal{B} \}$.
We define two types of \emph{ignorable branches:}
\begin{enumerate}
\item  $\mathcal{B}$ is an \emph{ignorable type-1 branch} if  $\mathcal{B}$ is completed and there is an atom $a$ such that $F\neg a$ is a node in $\mathcal{B}$, but $T a$ does not appear in $\mathcal{B}$.
\item  $\mathcal{B}$ is an \emph{ignorable type-2 branch} if there is another branch $\mathcal{B}'$ in the tableau that is completed but not $T$-closed, such that $\mathrm{At}(\mathcal{B}')\subset\mathrm{At}(\mathcal{B})$.
\end{enumerate}
%We write $\Gamma\vdash_{\OTAB} \Delta$ when for $\Gamma\vdash_M \Delta$ there is an Olivetti tableau for which every branch is closed or ignorable.

\begin{theorem}\textbf{(Theorem 2 in \cite{Oli92})} The  sequent $\Gamma\vdash_M \Delta$ is true if and only if there is an \OTAB tableau in which every branch is closed or ignorable.
\end{theorem}

\section{Simulating \OTAB by \MLK}
\label{sec:simulation-OTAB-MLK}

We will work towards a simulation of the tableau system \OTAB by the sequent system \MLK. In preparation for this a few lemmas are needed. We also add more information to the nodes (this can all be done in polynomial time).
We start with a fact about \LK (for a proof see \cite{BC14-ECCC}).

\begin{lemma}\label{thm_shortLK}
 For sets of formulae $\Gamma, \Delta$  and disjoints sets of atoms  $\Sigma^{+}, \Sigma^{-}$ with $\VAR(\Gamma\cup \Delta)=\Sigma^{+}\sqcup \Sigma^{-}$ we can efficiently construct polynomial-size \LK-proofs of $\Sigma^{+}, \neg \Sigma^{-}, \Gamma\vdash \Delta$ when the sequent is true.
\end{lemma}

We also need to derive a way of weakening in \MLK, and we show this in the next lemma.

\begin{lemma}\label{thm_MLK_weak}
From a sequent $\Gamma\vdash_M \Delta$ with  non-empty $\Delta$ we can derive $\Gamma\vdash_M \Delta, \Sigma$ in a polynomial-size \MLK-proof for any set of formulae $\Sigma$.
\end{lemma}
\begin{proof}
We take $\delta\in\Delta$, and from the $\LK$-axiom we get $\delta\vdash\delta$. From weakening in \LK we obtain $\Gamma,\delta\vdash \Delta, \Sigma$. Using rule ($\vdash \vdash_M$) we obtain  $\Gamma,\delta\vdash_M \Delta, \Sigma$.  We then derive $\Gamma\vdash_M \Delta, \Sigma$ using the ($M$-cut) rule.
\qed
\end{proof}

The proof makes essential use of the (M-cut) rule. As a result \MLK is not complete without (M-cut); e.g.\ the sequent $\emptyset \vdash_M \neg a, \neg b$ cannot be derived. A discussion on cut elimination in \MLK is given in \cite{Oli92}.

%Note that these proofs can be efficiently constructed in polynomial size.

\begin{lemma}\label{thm_abshort}

Let $T\tau$ be an $\alpha$-type formula with $\alpha_1=T\tau_1$, $\alpha_2=T\tau_2$, and let $F\psi$ be an $\alpha$-type formula with $\alpha_1=F\psi_1$, $\alpha_2=F\psi_2$. Similarly, let $T\phi$ be a $\beta$-type formula with $\beta_1=T\phi_1$, $\beta_2=T\phi_2$, and let $F\chi$ be an $\beta$-type formula with $\beta_1=F\chi_1$, $\beta_2=F\chi_2$.

The following sequents all can be proved with polynomial-size \LK-proofs:
%\begin{itemize}
$\tau \vdash \tau_1\wedge \tau_2$,
$ \tau_1\wedge\tau_2 \vdash \tau$,
 $\psi \vdash \psi_1\vee\psi_2$,
 $\psi_1\vee\psi_2 \vdash \psi$,
 $\phi \vdash \phi_1\vee\phi_2$,
 $\phi_1\vee\phi_2 \vdash \phi$,
 $\chi \vdash \chi_1\wedge\chi_2$, and 
 $ \chi_1\wedge\chi_2 \vdash \chi$.
%\end{itemize}

\end{lemma}

The straightforward proof of this involves checking all cases, which we omit here.

We now annotate the nodes $u$ in an \OTAB tableau with three sets of formulae $A_u$, $B_u$, $C_u$ and a set of branches $D_u$. This information will later be used to construct sequents $A_u\vdash_M B_u, C_u$, which will form the skeleton of the eventual \MLK proof that simulates the \OTAB tableau. Intuitively, if we imagine following a branch when constructing the tableau, $A_u$ corresponds to the current unmarked $T$-formulae on the branch, while $B_u$ corresponds to the current unmarked $F$-formulae. $C_u$ contains global information on all the branches that minimise the ignorable type-2 branches in the subtree with root $u$. 
%The formulae $A_u$ and $B_u$ are constructed similarly, while $C_u$ requires more work and uses the sets $D_v$ for nodes $v$. 
The formal definition follows. We start with the definition of the formulae $A_u$ and $B_u$, which proceeds by induction on the construction of the tableau.

\begin{definition}\label{def_ab}
Nodes $u$ in the \OTAB tableau from the initial tableau are annotated with $A_u=\Gamma$ and $B_u=\Delta$. 

For the inductive step, consider the case that the extension rule (A) was used on node $u$ for the $\alpha$-type signed formula $\phi$. If $\phi=T\chi$ has $\alpha_1=T\chi_1$, $\alpha_2=T\chi_2$ then for the node $v$ labelled $\alpha_1$ and the node $w$ labelled $\alpha_2$,  $A_v=A_w=(\{\chi_1,\chi_2\}\cup A_u)\setminus\{\chi\}$ and $B_u=B_v=B_w$. If $\phi=F\chi$ has $\alpha_1=F\chi_1$, $\alpha_2=F\chi_2$ then for the node $v$ labelled $\alpha_1$ and the node $w$ labelled $\alpha_2$,  $A_u=A_v=A_w$ and $B_v=B_w=(\{\chi_1,\chi_2\}\cup B_u)\setminus\{\chi\}$.

Consider now the case that the branching rule (B) was used on node $u$ for the $\beta$-type signed formula $\phi$. If $\phi=T\chi$ has $\beta_1=T\chi_1$, $\beta_2=T\chi_2$ then for the node $v$ labelled $\beta_1$ and the node $w$ labelled $\beta_2$, $A_v=(\{\chi_1\}\cup A_u)\setminus\{\chi\}, A_w=(\{\chi_2\}\cup A_u)\setminus\{\chi\}$ and $ B_v=B_w=B_u$. If $\phi=F\chi$ has $\beta_1=F\chi_1$, $\beta_2=F\chi_2$ then for the node $v$ labelled $\beta_1$ and the node $w$ labelled $\beta_2$, $B_v=(\{\chi_1\}\cup B_u)\setminus\{\chi\}, B_w=(\{\chi_2\}\cup B_u)\setminus\{\chi\}$ and $A_v=A_w=A_u$.
\end{definition}

For each ignorable type-2 branch $\mathcal{B}$ we can find another branch $\mathcal{B'}$,  which is not ignorable type-2 and such that $\mathrm{At}(\mathcal{B}')\subset\mathrm{At}(\mathcal{B})$. The definition of ignorable type-2 might just refer to another ignorable type-2 branch, but eventually --- since the tableau is finite --- we reach a branch $\mathcal{B'}$,  which is not ignorable type-2.  There could be several such branches, and we will denote the left-most such branch $\mathcal{B'}$ as $\theta(\mathcal{B})$.

We are now going to construct sets $C_u$ and $D_u$. The set $D_u$ contains some information on type-2 ignorable branches. Let $u$ be a node, which is the root of a sub-tableau $T$, and consider the set $I$ of all type-2 ignorable branches that go through $T$. Now intuitively, $D_u$ is defined as the set of all branches from $\theta(I)$ that are outside of $T$.
The set $C_u$ is then defined from $D_u$ as $C_u=\{\bigwedge_{p\in\mathrm{At}(\theta(\mathcal{B}))} p \mid \mathcal{B} \in D_u\}$. The formal constructions of $C_u$ and $D_u$ are below. Unlike $A_u$ and $B_u$, which are constructed inductively from the root of the tableau, the sets $C_u$ and $D_u$ are constructed inductively from the leaves to the root, by reversing the branching procedure.  

\begin{definition}\label{def_cd}
 For an ignorable type-2 branch $\mathcal{B}$ the end node $u$ is annotated by the singleton sets $C_u=\{\bigwedge_{p\in\mathrm{At}(\mathcal{\theta(B)})} p \}$ and $D_u=\{\theta(\mathcal{B})\}$; for other leaves $C_u=D_u=\emptyset$.

Inductively, we define:
\begin{itemize}

\item For a node $u$ with only one child $v$, we set $D_u=D_v$ and $C_u=C_v$.

\item For a node $u$ with two children $v$ and $w$, we set $D_u=(D_v\setminus\{\mathcal{B}\mid w\in\mathcal{B}\})\cup(D_w\setminus\{\mathcal{B} \mid v\in\mathcal{B}\})$ and $C_u=\{\bigwedge_{p\in\mathrm{At}(\theta(\mathcal{B}))} p \mid \mathcal{B} \in D_u\}$.
\end{itemize}

For each binary node $u$ with children $v$, $w$ we specify two extra sets. We set $E_u=(D_v\cup D_w) \setminus D_u$, and from this we can construct the set of formulae $F_u=\{\bigwedge_{p\in\mathrm{At}(\mathcal{B})} p \mid \mathcal{B} \in E_u\}$. We let $\omega=\bigvee F_u$.
\end{definition}

We now prepare the simulation result with a couple of lemmas.

\begin{lemma}\label{thm_leaf2}
Let $\mathcal{B}$ be a branch in an \OTAB tableau ending in leaf $u$. Then $A_u$ is the set of all unmarked $T$-formulae on $\mathcal{B}$ (with the sign $T$ removed). Likewise $B_u$ is the set of all unmarked $F$-formulae on $\mathcal{B}$ (with the sign $F$ removed). 
\end{lemma}
\begin{proof}
We will verify this for $T$-formulae, the argument is the same for $F$-formulae.
If $T\phi$ at node $v$ is an unmarked formula on branch $\mathcal{B}$ then $\phi$ has been added to $A_v$, regardless of which extension rule is used and cannot be removed at any node unless it is marked. Therefore, if $u$ is the leaf of the branch, we have $\phi\in A_u$. If $T\phi$ is marked then it is removed (in the inductive step in the construction in Definition~\ref{def_ab}) and is not present in $A_u$. $F$-formulae do not appear in $A_u$.
\qed
\end{proof}

\begin{lemma}\label{thm_leaf1} 
Let $\mathcal{B}$ be a branch in an \OTAB tableau.
\begin{enumerate}
\item Assume that $T\phi$ appears on the branch $\mathcal{B}$, and let $A(\mathcal{B})$ be the set of unmarked $T$-formulae on $\mathcal{B}$ (with the sign $T$ removed). Then $A(\mathcal{B})\vdash \phi$ can be derived in a polynomial-size \LK-proof.
\item Assume that $F(\phi)$ appears on the branch $\mathcal{B}$, and let $B(\mathcal{B})$ be the set of unmarked $F$-formulae on $\mathcal{B}$ (with the sign $F$ removed). Then $\phi\vdash B(\mathcal{B})$ can be derived in a polynomial-size \LK-proof.
\end{enumerate}
\end{lemma}
\begin{proof}
We prove the two claims by induction on the number of branching rules (A) and extension rules (B) that have been applied on the path to the node. We start with the proof of the first item.

%\begin{enumerate}

%\item 
\textbf{Induction Hypothesis} (on the number of applications of rules (A) and (B) on the node labelled $T\phi$): For a node labelled $T\phi$ on branch $\mathcal{B}$, we can derive $A(\mathcal{B})\vdash \phi$ in a polynomial-size \LK-proof (in the size of the tableau).

\textbf{Base Case} ($T\phi$ is unmarked): The \LK  axiom $\phi\vdash\phi$ can be used and then weakening to obtain $A(\mathcal{B})\vdash \phi$.

\textbf{Inductive Step:} If $T\phi$ is a marked $\alpha$-type formula, then both $\alpha_1= T\phi_1$ and $\alpha_2=T\phi_2$ appear on the branch. By the induction hypothesis we derive $A(\mathcal{B})\vdash\phi_1$, $A(\mathcal{B})\vdash\phi_2$ in polynomial-size proofs, hence we can derive $A(\mathcal{B})\vdash\phi_1\wedge \phi_2$ in a polynomial-size proof (we are bounded in total number of proof subtrees by the numbers of nodes in our branch). We then have $\phi_1\wedge \phi_2\vdash \phi$ using Lemma~\ref{thm_abshort}. Using the cut rule we can derive $A(\mathcal{B})\vdash \phi$.

If $T\phi$ is a $\beta$-type formula and is marked, then the branch must contain $\beta_1= T\phi_1$ or $\beta_2=T\phi_2$. Without loss of generality we can assume that $\beta_1= T\phi_1$ appears on the branch. By the induction hypothesis $A(\mathcal{B})\vdash\phi_1$, therefore we can derive $A(\mathcal{B})\vdash\phi_1\vee \phi_2$ since it is a $\beta$-type formula and derive $\phi_1\vee \phi_2\vdash \phi$  with Lemma~\ref{thm_abshort}. Then using the cut rule we derive $A(\mathcal{B})\vdash \phi$.

%\item
The second item is again shown by induction.

\textbf{Induction Hypothesis} (on the number of applications of rules (A) and (B) on the node labelled $F\phi$): For a node labelled $F\phi$ on branch $\mathcal{B}$, we can derive $\phi\vdash B(\mathcal{B})$ in a polynomial-size \LK-proof (in the size of the tableau).

 \textbf{Base Case} ($F\phi$ is unmarked): The \LK axiom $\phi\vdash\phi$ can be used and then weakened to $\phi\vdash B(\mathcal{B})$.

\textbf{Inductive Step:} If $F\phi$ is a marked $\alpha$-type formula, then both $\alpha_1= F\phi_1$ and $\alpha_2=F\phi_2$ appear on the branch. Since by the inductive hypothesis $\phi_1\vdash B(\mathcal{B})$ and $\phi_2\vdash B(\mathcal{B})$, we can derive $\phi_1\vee \phi_2\vdash B(\mathcal{B})$ in a polynomial-size proof. We then have $\phi\vdash \phi_1\vee\phi_2$ using Lemma~\ref{thm_abshort}. Using the cut rule we can derive $\phi\vdash B(\mathcal{B})$.

If $F\phi$ is a $\beta$-type formula and is marked, then the branch must contain $\beta_1= F\phi_1$ or $\beta_2=F\phi_2$. Without loss of generality we can assume $\beta_1= F\phi_1$ appears on the branch. By the induction hypothesis $\phi_1\vdash B(\mathcal{B})$, therefore we can derive $\phi_1\wedge\phi_2\vdash B(\mathcal{B})$ since it is a $\beta$-type formula and derive $\phi\vdash \phi_1\wedge\phi_2$ with Lemma~\ref{thm_abshort}. Using the cut rule we derive $\phi\vdash B(\mathcal{B})$.
\qed
%\end{enumerate}
\end{proof}

\begin{lemma}\label{thm_atsat}
Let $\mathcal{B}$ be a branch, which is completed but not $T$-closed. For any node $u$ on  $\mathcal{B}$, the model $\mathrm{At}(\mathcal{B})$ satisfies $A_u$.
\end{lemma}
\begin{proof}
We prove the lemma by induction on the height of the subtree with root $u$.

\textbf{Base Case} ($u$ is a leaf): By Lemma~\ref{thm_leaf2}, $A_u$ is the set of all unmarked $T$-formulae on $\mathcal{B}$. But these are all literals as $\mathcal{B}$ is completed, and hence the subset of positive atoms is equal to $\mathrm{At}(\mathcal{B})$.

\textbf{Inductive step}: If $u$ is of extension type (A) with child node $v$ then the models of $A_u$ are exactly the same as the models of $A_v$. This is true for all $\alpha$-type formulae. For example, if the extension process (A) was used on formula $T(\psi\wedge\chi)$ and the node $v$ was labelled $T\psi$ then $A_v=\{\psi, \chi\}\cup A_u\setminus\{\psi\wedge\chi\}$ and this has the same models as $A_u$. By the induction hypothesis, $\mathrm{At}(\mathcal{B})\models A_v$ and hence $\mathrm{At}(\mathcal{B})\models A_u$.

If $u$ is of branch type (B) with children $v$ and $w$ then $\mathrm{At}(\mathcal{B})\models A_v$ and $\mathrm{At}(\mathcal{B})\models A_w$. The argument works similarly for all $\beta$-type formulae; for example, if the extension process was using formula $T(\psi\vee\chi)$ and $v$ is labelled $T\psi$ and $w$ is labelled $T\chi$, then $A_u=(\{\psi\vee\chi\}\cup A_v)\setminus\{\psi\}$. Hence $\mathrm{At}(\mathcal{B})\models A_v$ implies $\mathrm{At}(\mathcal{B})\models A_u$.
\qed
\end{proof}

We now approach the simulation result (Theorem~\ref{thm_MLK>OTAB}) and start to construct \MLK proofs. 
%in Theorem~\ref{thm_MLK>OTAB}. Lemma~\ref{thm_omegagamma} is used to prove Lemma~\ref{thm_omegaB}. 
For the next two lemmas, we fix an \OTAB tableau of size $k$ and use the notation from Definitions~\ref{def_ab} and \ref{def_cd} (recall in particular the definition of  $\omega$ at the end of Definition~\ref{def_cd}).

\begin{lemma}\label{thm_omegagamma}
 There is a polynomial $q$ such that for every binary node $u$, every proper subset $A' \subset A_u$ and every $\gamma\in A_u\setminus A'$ we can construct an \MLK-proof of $A',\omega \vdash_M \gamma $ of size at most $q(k)$.
\end{lemma}
\begin{proof}
\textbf{Induction Hypothesis} (on the number of formulae of $A_u$ used in the antecedent: $|A'|$): We  can find a $q(k)$-size \MLK proof containing all sequents  $A',\omega \vdash_M \gamma$ for every $\gamma\in A_u\setminus A'$ .

\textbf{Base Case} (when $A'$ is empty):
For the base case we aim to prove $\omega\vdash_M \gamma$, and repeat this for every $\gamma$. We use two ingredients. Firstly, we need the sequent $\omega\vdash_M F_u, \gamma$ which is easy to prove using weakening and ($\vee\vdash$), since $\omega$ is a disjunction of the elements in $F_u$. Our second ingredient is a scheme of $\omega,\bigwedge_{p\in M} p \vdash_M \gamma $ for all the $\bigwedge_{p\in M} p$ in $F_u$, \ie $M=\mathrm{At}(\mathcal{B})$ for some $\mathcal{B}\in E_u$. With these we can repeatedly use (M-cut) on the first sequent for every element in $F_u$. We now show how to efficiently prove the sequents of the form $\omega,\bigwedge_{p\in M} p \vdash_M \gamma $. 

 For branch $\mathcal{B}\in E_u$, as $\mathrm{At}(\mathcal{B})$ is a model $M$ for $A_u$ by Lemma~\ref{thm_atsat}, $M\models\gamma$.
Since no atom  $a$ in $\VAR(\gamma)\setminus M$ appears positive in the set $M$ we can infer $M\vdash_M \neg a$ directly via $(\vdash_M)$. With rule ($\vdash_M\wedge$) we can derive $\bigwedge_{p\in M} p \vdash_M \bigwedge_{p\in \VAR(\gamma)\setminus M}\neg p$ in a polynomial-size proof.
Using ($\vdash$), ($\vdash\vee\bullet$), and ($\vdash\bullet\vee$) we can derive $\bigwedge_{p\in M} p\vdash \omega$. We then use these sequents in the proof below, denoting $\bigwedge_{p\in \VAR(\gamma)\setminus M}\neg p$ as $n(M)$:
\begin{prooftree}
\AxiomC{$\bigwedge_{p\in M} p\vdash \omega$}
\RightLabel{($\vdash\vdash_M$)}
\UnaryInfC{$\bigwedge_{p\in M} p\vdash_M \omega$}
\AxiomC{\hspace*{-0.5em}$\bigwedge_{p\in M} p \vdash_M n(M)$}
\RightLabel{($\bullet\vdash_M$)}
\BinaryInfC{$\omega, \bigwedge_{p\in M} p \vdash_M  n(M)$}
\end{prooftree}

From Lemma~\ref{thm_shortLK},
$M, \neg \VAR(\gamma)\setminus M \vdash \gamma$ can be derived in a polynomial-size proof. We use simple syntactic manipulation to change the antecedent into an equivalent conjunction and then weaken to derive $\omega, \bigwedge_{p\in M} p, \bigwedge_{p\in \VAR(\gamma)\setminus M} \neg p \vdash_M \gamma $ in a polynomial-size proof.
Then we use:

\begin{prooftree}
\AxiomC{$\omega,\bigwedge_{p\in M} p, n(M) \vdash_M \gamma$}
\AxiomC{$\omega, \bigwedge_{p\in M} p \vdash_M  n(M)$}
\RightLabel{(M-cut)}
\BinaryInfC{$\omega,\bigwedge_{p\in M} p \vdash_M \gamma $}
\end{prooftree}

\textbf{Inductive Step}: We look at proving $A', \gamma', \omega\vdash_M \gamma$, for every other $\gamma\in A_u\setminus A'$.  For each $\gamma$ we use two instances of the inductive hypothesis: $A',\omega \vdash_M \gamma$ and $A',\omega \vdash_M \gamma'$.
\begin{prooftree}
\AxiomC{$A',\omega \vdash_M \gamma'$}
\AxiomC{$A',\omega\vdash_M \gamma$}
\RightLabel{($\bullet\vdash_{M}$)}
\BinaryInfC{$A', \gamma', \omega\vdash_M \gamma$}
\end{prooftree}

Since we repeat this for every $\gamma$ we only add  $|(A_u\setminus A')\setminus \{\gamma\}|$ many lines in each inductive step and retain a polynomial bound.
\qed
\end{proof}

The previous lemma was an essential preparation for our next Lemma~\ref{thm_omegaB}, which in turn will be the crucial ingredient for the p-simulation in Theorem~\ref{thm_MLK>OTAB}.

\begin{lemma}\label{thm_omegaB}
There is a polynomial $q$ such for every binary node $u$ there is an \MLK-proof of $A_u,\omega \vdash B_u $ of size at most $q(k)$. 
\end{lemma}
\begin{proof}
\textbf{Induction Hypothesis} (on the number of formulae of $A_u$ used in the antecedent: $|A'|$): Let $A'\subseteq A_u$. There is a fixed polynomial $q$ such that $A', \omega \vdash B_u $ has an \MLK-proof of size at most $q(|\omega|)$.

\textbf{Base Case} (when $A'$ is empty):
We approach this very similarly as in the previous lemma.
Using weakening and ($\vee\vdash$), the sequent $\omega\vdash_M F_u, B_u$ can be derived in a polynomial-size proof. By repeated use of the cut rule on sequents of the form $\omega,\bigwedge_{p\in \mathrm{At}(\mathcal{B})} p \vdash_M B_u $ for $\mathcal{B}\in E_u$ the sequent $\omega\vdash_M B_u$ is derived. Now we only need to show that we can efficiently obtain $\omega,\bigwedge_{p\in M} p \vdash_M B_u $.

Consider branch $\mathcal{B}\in E_u$. As $\mathrm{At}(\mathcal{B})$ is a minimal model $M$ for $\Gamma$ by Lemma~\ref{thm_atsat}, this model
$M$ must satisfy $\Delta$ and given the limitations of the branching processes of $F$-labelled formulae, $B_u$ as well.

Similarly as in the base case of Lemma~\ref{thm_omegagamma} we can derive $\bigwedge_{p\in M} p \vdash_M \bigwedge_{p\in \VAR(B_u)\setminus M}\neg p$ and $\bigwedge_{p\in M} p\vdash \omega$ in a polynomial-size proof.
We then use these sequents in the proof below once again, denoting $\bigwedge_{p\in \VAR(\gamma)\setminus M}\neg p$ as $n(M)$.

\begin{prooftree}
\AxiomC{$\bigwedge_{p\in M} p\vdash \omega$}
\RightLabel{($\vdash\vdash_M$)}
\UnaryInfC{$\bigwedge_{p\in M} p\vdash_M \omega$}
\AxiomC{\hspace*{-0.5em}$\bigwedge_{p\in M} p \vdash_M n(M)$}
\RightLabel{($\bullet\vdash_M$)}
\BinaryInfC{$\omega, \bigwedge_{p\in M} p \vdash_M n(M)$}
\end{prooftree}

We can use $M$ satisfying $B_u$ to derive $\omega,\bigwedge_{p\in M} p, n(M) \vdash B_u$ in the same way as we derive $\omega, \bigwedge_{p\in M} p, \bigwedge_{p\in \VAR(\gamma)\setminus M} \neg p \vdash \gamma $ in Lemma~\ref{thm_omegagamma}.

{\small
\begin{prooftree}
\AxiomC{$\omega,\bigwedge_{p\in M} p, n(M)  \vdash_M B_u$}
\AxiomC{$\omega, \bigwedge_{p\in M} p \vdash_M n(M)$}
\RightLabel{(M-cut)}
\BinaryInfC{$\omega,\bigwedge_{p\in M} p \vdash_M B_u $}
\end{prooftree}
}

\textbf{Inductive Step}:
Assume that $A',\omega\vdash_M B_u$ has already been derived.
Let $\gamma\in A_u \setminus A'$. We use Lemma~\ref{thm_omegagamma} to get a short proof of $A',\omega \vdash_M \gamma$. One application of rule $(\bullet\vdash_{M})$
\begin{prooftree}
\AxiomC{$A',\omega\vdash_M B_u$}
\AxiomC{$A',\omega \vdash_M \gamma$}
\RightLabel{($\bullet\vdash_{M}$)}
\BinaryInfC{$A', \gamma,\omega \vdash_M B_u$}
\end{prooftree} 
finishes the proof.
\qed
\end{proof}

\begin{theorem}\label{thm_MLK>OTAB}
\MLK p-simulates \OTAB.
\end{theorem}
\begin{proof}
\textbf{Induction Hypothesis} (on the height of the subtree with root $u$): For node $u$, we can derive $A_u\vdash_M B_u, C_u$ in \MLK in polynomial size (in the full tableau).

\textbf{Base Case} ($u$ is a leaf): If the branch is $T$-closed, then by Lemma~\ref{thm_leaf1}, for some formula $\phi$ we can derive $A_u\vdash \phi$ and $A_u\vdash\neg \phi$. Hence $A_u\vdash \phi\wedge\neg\phi$ can be derived and with $\phi\wedge\neg\phi\vdash$ and the cut rule we can derive $A_u\vdash $ in a polynomial-size proof. By weakening and using ($\vdash\vdash_M$) we can derive $A_u\vdash_M B_u$ in polynomial size as required.

If the branch is $F$-closed, then by Lemma~\ref{thm_leaf1}, for some formula $\phi$ we can derive $\phi\vdash B_u$ and $\neg \phi\vdash B_u$. Hence 
$\phi\vee\neg\phi\vdash B_u$ can be derived and with $\vdash\phi\vee\neg\phi$ and the cut rule we can derive $\vdash B_u $ in a polynomial-size proof. By weakening and using ($\vdash\vdash_M$) we can derive $A_u\vdash_M B_u$ in polynomial size.

If the branch is $TF$-closed, then by Lemma~\ref{thm_leaf1}, for some formula $\phi$ we can derive $A_u\vdash \phi$ and $\phi\vdash B_u$. Hence via the cut rule and using ($\vdash\vdash_M$) we can derive $A_u\vdash_M B_u$ in polynomial size as required.

If the branch is ignorable type-1 then the branch is completed. Therefore $A_u$ is a set of atoms and there is some atom $a\notin A_u$ such that  $\neg a\in B_u$. It therefore follows that $A_u\vdash_M \neg a$ can be derived as an axiom using the ($\vdash_M$) rule. We then use Lemma~\ref{thm_MLK_weak} to derive $A_u \vdash_M B_u $ in polynomial size.

If the branch is ignorable type-2 then $p\in \mathrm{At}(\theta(\mathcal{B}))$ implies $p\in A_u$. Since $C_u=\{\bigwedge_{p\in\mathrm{At}(\theta(\mathcal{B}))} p \}$ we can find a short proof  of $A_u\vdash C_u$ using ($\vdash\wedge$).

\textbf{Inductive Step}: The inductive step splits into four cases according to which extension or branching rule is used on node $u$.

\emph{Case 1.} Extension rule (A) is used on node $u$ for formula $T\phi$ with resulting nodes $v$ and $w$ labelled $T\phi_1$, $T\phi_2$, respectively. 

\begin{prooftree}
\AxiomC{}
\UnaryInfC{$\phi_1\vdash\phi_1$}
\RightLabel{($\bullet\vdash$)}
\UnaryInfC{$\phi_1, \phi_2\vdash\phi_1 $}
\AxiomC{}
\UnaryInfC{$\phi_2\vdash\phi_2$}
\RightLabel{($\bullet\vdash$)}
\UnaryInfC{$\phi_1, \phi_2\vdash\phi_2 $}
\RightLabel{($\vdash\wedge$)}
\BinaryInfC{$\phi_1, \phi_2 \vdash \phi_1\wedge\phi_2$}
\end{prooftree}
Since we are extending the branch on an $\alpha$-type formula signed with $T$, we can find a short proof of $\phi_1\wedge\phi_2\vdash \phi$ using Lemma~\ref{thm_abshort}.
%Using Lemma~\ref{thm_shortLK}, 
Together with $\phi_1, \phi_2 \vdash \phi_1\wedge\phi_2$ shown above we derive:

\begin{prooftree}
\AxiomC{$\phi_1, \phi_2 \vdash \phi_1\wedge\phi_2$}
\AxiomC{$\phi_1\wedge\phi_2 \vdash \phi$}
\RightLabel{(cut)}
\BinaryInfC{$\phi_1, \phi_2 \vdash \phi$}
\end{prooftree}

By definition we have $\phi_1,\phi_2\in A_v$, and then by weakening $\phi_1, \phi_2 \vdash \phi$ we obtain $A_v\vdash \phi$. By Definitions~\ref{def_ab} and \ref{def_cd}, $B_v=B_u$ and likewise $C_u=C_v$. Hence $A_v\vdash_M B_u, C_u$ is available by the induction hypothesis. From this we get:

\begin{prooftree}
\AxiomC{$A_v\vdash \phi$}
\RightLabel{($\vdash\vdash_M$)}
\UnaryInfC{$A_v\vdash_M \phi$}
\AxiomC{$A_v\vdash_M B_u, C_u$}
\RightLabel{($\bullet \vdash_M$)}
\BinaryInfC{$A_v, \phi \vdash_M B_u, C_u$}
\end{prooftree}

$A_u\vdash \phi_1$ and $A_u\vdash \phi_2$ also have short proofs from weakening axioms. These can be used to cut out $\phi_1,\phi_2$ from the antecedent of $A_v, \phi \vdash_M B_u, C_u$ resulting in $A_u \vdash_M B_u, C_u$ as required.

\emph{Case 2.}  
Extension rule (A) is used on node $u$ for formula $F\phi$ with resulting nodes $v$ and $w$ labelled $F\phi_1$, $F\phi_2$, respectively. We can find short proofs of $A_u,\phi_1 \vdash \phi_1\vee\phi_2$, $A_u,\phi_2 \vdash \phi_1\vee\phi_2$ using axioms, weakening and the rules ($\vdash \bullet \vee$), ($\vdash  \vee \bullet$). Similarly as in the last case, we have $A_v=A_u$ and likewise $C_u=C_v$. Therefore, by induction hypothesis $A_u \vdash_M B_v, C_u$ is available with a short proof.

\begin{prooftree}
\AxiomC{$A_u \vdash_M B_v, C_u$}
\AxiomC{$A_u,\phi_1 \vdash \phi_1\vee\phi_2$}
\RightLabel{($\vdash\vdash_M$)}
\UnaryInfC{$A_u,\phi_1 \vdash_M \phi_1\vee\phi_2$}
\RightLabel{(M-cut)}
\BinaryInfC{$A_u \vdash_M B_v\setminus\{\phi_1\}, \phi_1\vee\phi_2, C_u$}
\end{prooftree}

We can do the same trick with $\phi_2$:
{\footnotesize
\begin{prooftree}
\AxiomC{$A_u \vdash_M B_v\setminus\{\phi_1\}, \phi_1\vee\phi_2, C_u$}
\AxiomC{$A_u,\phi_2 \vdash \phi_1\vee\phi_2$}
\RightLabel{($\vdash\vdash_M$)}
\UnaryInfC{$A_u,\phi_2 \vdash_M \phi_1\vee\phi_2$}
\RightLabel{(M-cut)}
\BinaryInfC{$A_u \vdash_M B_u\setminus\{\phi\}, \phi_1\vee\phi_2, C_u$}
\end{prooftree}
}

Since $F\phi$ is an $\alpha$-type formula, then $\phi_1\vee\phi_2\vdash\phi$ by Lemma~\ref{thm_abshort}, and by weakening $A_u, \phi_1\vee\phi_2\vdash\phi$.
The derivation is the finished by:
{\footnotesize
\begin{prooftree}
\AxiomC{$A_u \vdash_M B_u\setminus\{\phi\}, \phi_1\vee\phi_2, C_u$}
\AxiomC{$A_u, \phi_1\vee\phi_2\vdash\phi$}
\RightLabel{($\vdash\vdash_M$)}
\UnaryInfC{$A_u, \phi_1\vee\phi_2\vdash_M\phi$}
\RightLabel{(M-cut)}
\BinaryInfC{$A_u \vdash_M B_u, C_u$}
\end{prooftree}
}

\emph{Case 3.} 
Branching rule (B) is used on node $u$ for formula $T\phi$ with children $v$ and $w$ labelled $T\phi_1$, $T\phi_2$, respectively. The sequents $A_v\vdash_M B_u, C_v$ and $A_w\vdash_M B_u, C_w$ are available from the induction hypothesis. 

$A_v\vdash_M B_u, C_u, F_u$ and $A_w\vdash_M B_u, C_u, F_u$ can be derived via weakening by Lemma~\ref{thm_MLK_weak}. From these sequents, simple manipulation through classical logic and the cut rule gives us $A_v\vdash_M B_u, C_u, \omega$ and $A_w\vdash_M B_u, C_u, \omega$. Using the rule $(\vee\vdash_M)$ we obtain $A_u\setminus\{\phi\}, \phi_1\vee\phi_2 \vdash_M B_u, C_u, \omega$. Since $\phi\in A_u$, from Lemma~\ref{thm_abshort} we derive
$\phi\vdash \phi_1\vee\phi_2$ and $\phi_1\vee\phi_2\vdash \phi$ in polynomial size.
Weakening derives $A_u \vdash\phi_1\vee\phi_2$ and $A_u\setminus\{\phi\}, \phi_1\vee\phi_2 \vdash \phi$.
From these we derive:

{\scriptsize
\begin{prooftree}
\AxiomC{$A_u\setminus\{\phi\}, \phi_1\vee\phi_2 \vdash_M B_u, C_u, \omega$}
\AxiomC{$A_u\setminus\{\phi\}, \phi_1\vee\phi_2 \vdash \phi$}
\RightLabel{($\vdash\vdash_M$)}
\UnaryInfC{$A_u\setminus\{\phi\}, \phi_1\vee\phi_2 \vdash_M \phi$}
\RightLabel{($\bullet\vdash_M$)}
\BinaryInfC{$A_u, \phi_1\vee\phi_2\vdash_M B_u, C_u, \omega $}
\end{prooftree}
}

\begin{prooftree}
\AxiomC{$A_u, \phi_1\vee\phi_2\vdash_M B_u, C_u, \omega $}
\AxiomC{$A_u \vdash\phi_1\vee\phi_2$}
\RightLabel{($\vdash\vdash_M$)}
\UnaryInfC{$A_u \vdash_M\phi_1\vee\phi_2$}
\RightLabel{(M-cut)}
\BinaryInfC{$A_u \vdash_M B_u, C_u, \omega $}
\end{prooftree}

From Lemma~\ref{thm_omegaB}, $A_u, \omega\vdash_M B_u, C_u $ has a polynomial size proof. We can then finish the derivation with a cut:

\begin{prooftree}
\AxiomC{$A_u, \omega \vdash_M B_u$}
\AxiomC{$A_u\vdash_M B_u, C_u, \omega $}
\RightLabel{(M-cut)}
\BinaryInfC{$A_u\vdash_M B_u, C_u$}
\end{prooftree}

\emph{Case 4.} 
Branching rule (B) is used on node $u$ for formula $F\phi$ with children $v$ and $w$ labelled $F\phi_1$, $F\phi_2$, respectively. The sequents $A_u\vdash_M B_v, C_v$ and $A_u\vdash_M B_w, C_w$ are available from the induction hypothesis.

From these two sequents we obtain via weakening $A_u\vdash_M B_v, C_u, F_u$ and $A_u\vdash_M B_w, C_u, F_u$. We can turn $F_u$ into the disjunction of its elements by simple manipulation through classical logic and the cut rule  and derive $A_u\vdash_M B_v, C_u, \omega$ and $A_u\vdash_M B_w, C_u, \omega$. Using the rule $(\vdash_M\wedge)$ we obtain $ A_u\vdash_M  B_u\setminus\{\phi\}, \phi_1\wedge\phi_2, C_u, \omega$. 
Since $\phi_1\wedge\phi_2\vdash \phi$ by Lemma~\ref{thm_abshort}, we derive by
weakening $A_u, \phi_1\wedge\phi_2 \vdash \phi$. We then continue:

{\footnotesize
\begin{prooftree}
\AxiomC{$A_u\vdash_M B_u\setminus\{\phi\}, \phi_1\wedge\phi_2, C_u, \omega$}
\AxiomC{$A_u, \phi_1\wedge\phi_2 \vdash \phi$}
\RightLabel{($\vdash\vdash_M$)}
\UnaryInfC{$A_u, \phi_1\wedge\phi_2 \vdash_M \phi$}
\RightLabel{(M-cut)}
\BinaryInfC{$A_u \vdash_M B_u, C_u, \omega $}
\end{prooftree}
}

From Lemma~\ref{thm_omegaB}, $A_u, \omega\vdash_M B_u, C_u $ has a polynomial-size proof.

\begin{prooftree}
\AxiomC{$A_u, \omega \vdash_M B_u$}
\AxiomC{$A_u\vdash_M B_u, C_u, \omega $}
\RightLabel{(M-cut)}
\BinaryInfC{$A_u\vdash_M B_u, C_u$}
\end{prooftree}
This completes the proof of the induction.

From this induction, the theorem can be derived as follows.
The induction hypothesis applied to the root $u$ of the tableau  gives polynomial-size \MLK proofs of $A_u \vdash_M B_u, C_u$. By definition $A_u= \Gamma$ and $B_u= \Delta$. 
Finally, $C_u=D_u=\emptyset$, because for every ignorable type-2 branch $\mathcal{B}$, the branch $\theta(\mathcal{B})$ is inside the tableau. 

Since all our steps are constructive we prove a p-simulation.
\qed
\end{proof}

\section{Separating \OTAB and \MLK}
\label{sec:separation-OTAB-MLK}

In the previous section we showed that \MLK p-simulates \OTAB. Here we prove that the two systems are in fact exponentially separated.

\begin{lemma}\label{thm_incon}
In every \OTAB tableau for $\Gamma \vdash_M \Delta$ with inconsistent $\Gamma$, any completed branch is $T$-closed.
\end{lemma}
\begin{proof}
If a branch $\mathcal{B}$ is completed but not $T$-closed, then via Lemma~\ref{thm_atsat},  $\mathrm{At}(\mathcal{B})$ is a model for all initial $T$-formulae. But these form an inconsistent set. \qed
\end{proof}

\begin{theorem} \label{thm:sep-TAB-MLK}
\OTAB does not simulate \MLK.
\end{theorem}
\begin{proof}
We consider Smullyan's \emph{analytic tableaux} \cite{Smu68}, and use the hard sets of inconsistent formulae in \cite{MD92}.

For each natural number $n>0$ we use variables $p_1, \dots, p_n$. Let $H_n$ be the set of all $2^n$ clauses of length $n$ over these variables (we exclude tautological clauses) and define $\phi_n=\bigwedge H_n$. Since every model must contradict one of these clauses, $\phi_n$ is inconsistent. We now consider the sequents $\phi_n \vdash_M$.

Since classical entailment is included in minimal entailment there must also be an \OTAB tableau for these formulae. Every type-1 ignorable branch in the \OTAB tableau is completed and therefore also $T$-closed by Lemma~\ref{thm_incon}. The tableau cannot contain any  type-2 ignorable branches as every completed branch is $T$-closed. Hence the \OTAB tableaux for $\phi_n \vdash_M$ are in fact analytic tableaux and have $n!$ many branches by Proposition~1 from \cite{MD92}.

Since the examples are easy for truth tables \cite{MD92}, they are also easy for $\LK$ and the rule ($\vdash\vdash_M$) completes a polynomial-size proof for them in \MLK. 
\qed
\end{proof}

\section{Conclusion}
\label{sec:conclusion}

In this paper we have clarified the relationship between the proof systems \OTAB and \MLK for minimal entailment.  While cut-free sequent calculi typically have the same proof complexity as tableau systems, \MLK is not complete without M-cut \cite{Oli92}, and also our translation uses M-cut in an essential way (however, we can eliminate \LK-cut).

We conclude by mentioning that there are further proof systems for minimal entailment and circumscription, which have been recently analysed from a proof-complexity perspective \cite{BC14-ECCC}. In particular, Niemel\"{a} \shortcite{Nie96} introduced a tableau system \NTAB for minimal entailment for clausal formulas, and Bonatti and Olivetti \shortcite{BO02} defined an analytic sequent calculus \CIRC for circumscription. Building on initial results from \cite{BO02} we prove in \cite{BC14-ECCC} that $\NTAB \leq_p \CIRC \leq_p \MLK$ is a chain of proof systems of strictly increasing strength, \ie in addition to the p-simulations we obtain separations between the proof systems. 

Combining the results of \cite{BC14-ECCC} and the present paper, the full picture of the simulation order of proof systems for minimal entailment emerges. In terms of proof size, \MLK is the best proof system as it p-simulates all other known proof systems. However, for a  complete understanding of the simulation order  some problems are still open.
While the separation between \OTAB and \MLK from Theorem~\ref{thm:sep-TAB-MLK} can be straightforwardly adapted to show that \OTAB also does not simulate \CIRC, we leave open whether the reverse simulation holds. Likewise, the relationship between the two tableau systems \OTAB and \NTAB is not clear.

It is also interesting to compare our results to the complexity of theorem proving procedures in other non-monotonic logics as default logic \cite{BMMTV11} and autoepistemic logic \cite{Bey13}; cf.\ also \cite{ET01} for results on proof complexity in the first-order versions of some of these systems. In particular, \cite{BMMTV11} and \cite{Bey13} show very close connections between proof lengths in some sequent systems for default and autoepistemic logic and proof lengths of classical \LK, for which any non-trivial lower bounds are a major outstanding problem. It would be interesting to know if a similar relation also holds between \MLK and \LK.

%\bibliographystyle{aaai}
%\bibliography{compl} 

\end{document}